\documentclass[11pt]{article}
\usepackage{amsmath, amsthm, amsfonts}
\usepackage[margin=2.5cm]{geometry}
\usepackage{amssymb}
\usepackage[cp1250]{inputenc}
\newtheorem{thm}{Theorem}[section]

\theoremstyle{definition}

\theoremstyle{remark}

\title{A new approach to measurement in quantum tomography}

\author{Artur Czerwi{\'n}ski\\
E-mail: aczerwin@fizyka.umk.pl\\
  \small 1. Institute of Physics\\
  \small Nicolaus Copernicus University\\
  \small 87-100 Toru{\'n}\\
	\small 2. Center for Theoretical Physics \\
\small Polish Academy of Sciences \\
\small 02-668 Warszawa\\
}

\begin{document}
\maketitle

\abstract{In this article we propose a new approach to quantum measurement in reference to the stroboscopic tomography. Generally, in the stroboscopic approach it is assumed that the information about the quantum system is encoded in the mean values of certain hermitian operators $Q_1, ..., Q_r$ and each of them can be measured more than once. The main goal of the stroboscopic tomography is to determine when one can reconstruct the initial density matrix $\rho(0)$ on the basis of the measurement results $\langle Q_i \rangle _{t_j}$. In this paper we propose to treat every complex matrix as a measurable operator. This generalized approach to quantum measurement may bring some improvement into the models of stroboscopic tomography.
}
\section{Introduction}

In this paper by $\mathcal{H}$ we shall denote the Hilbert space and we shall assume that $dim \mathcal{H} = n < \infty $. By $B(\mathcal{H})$ we shall denote the complex vector space of all bounded linear operators in $\mathcal{H}$. The space $B(\mathcal{H})$ is isomorphic with the space of all complex matrices that shall be represented by $\mathbb{M}_n (\mathbb{C})$. Finally, $B_* (\mathcal{H})$ shall refer to the real vector space of all hermitian (self-adjoint) operators on $\mathcal{H}$. The elements of $B_* (\mathcal{H})$ shall be called observables.

The term quantum tomography refers to a wide variety of methods and approaches which aim to reconstruct the accurate representation of a quantum system by conducting a series a measurements. Among many different approaches to quantum tomography one can especially mention the so-called static model of tomography, which requires $n^2 - 1$ measurements of different observables taken at time instant $t=0$ (see more in \cite{altepeter04,alicki87,genki03}). A paper published in 2011 initiated another approach to quantum tomography which is based on weak measurement. The paper revealed that the wave function of a pure state can be measured in a direct way \cite{bamber}. Further papers proved that this approach can be generalized also for mixed state identification \cite{wu}.

In this paper we follow yet another approach to quantum tomography - the so-called stroboscopic tomography which originated in 1983 in the article \cite{jam83}. Subsequently, the approach was developed in other papers, such as \cite{jam00,jam04}. The assumption that is at the very foundation of this method claims that the evolution of an open quantum system can be expressed by a master equation of the form
\begin{equation}\label{eq:kossak}
\dot{\rho} = \mathbb{L} \rho,
\end{equation}
where the operator $\mathbb{L}$ is called the generator of evolution and its most general form have been introduced in \cite{gorini76}. In order to determine the initial density matrix $\rho(0)$ one assumes to have a set of identically prepared quantum systems which evolve according to the master equation with generator $\mathbb{L}$. Each system can be measured only once, because any measurement, generally, influences the state. 

The other underlying assumption connected with the stroboscopic approach is that the knowledge about the quantum system is provided by mean values of certain observables $\{Q_1, ..., Q_r\}$ (obviously $Q_i ^* = Q_i$) such that $r < n^2 -1$. These mean values are mathematically expressed as
\begin{equation}
\langle Q_i \rangle _t = Tr(Q_i \rho(t))
\end{equation}
and are assumed to be achievable from an experiment. If we additionally assume that the knowledge about the evolution enables us to perform measurements at different time instants $t_1, ..., t_g$, we get from an experiment a matrix of data $[\langle Q_i \rangle _{t_j}]$, where $i=1,...,r$ and $j=1,...,g$. The fundamental question of the stroboscopic tomography that one asks is whether the matrix of experimental data is sufficient to reconstruct the initial density matrix $\rho(0)$. Other problems relate to the minimal number of observables and time instants, the properties of the observables and the choice of time instants. In general the conditions under which it is possible to reconstruct the initial state have been determined and can be found in \cite{jam83,jam00,jam04}.

Compared with the static model of tomography, the stroboscopic approach makes it possible to decrease significantly the number of different observables that are necessary to perform quantum tomography. From experimantal point of view it means that in the static model one needs to preprare $n^2 -1$ different experimental systems (e.g. for $dim \mathcal{H} = 4$ one would need to measure $15$ different quantities), which seems rather unrealistic. Therefore, the stroboscopic approach appears to have an advantage over the static model as it aims to minimalize the number of different observables.

\section{Generalized observables and measurement results}

According to one of the most fundamental concepts of quantum mechanics to every physical quantity we can assign a hermitian operator which is called an observable. Thus when talking about measurements in the context of the stroboscopic tomography we consider mean values of certain hermitian operators \cite{jam83}. The main goal of this section is to prove that this approach to measurement can be generalized in such a way that any complex matrix $A\in \mathbb{M}_n (\mathbb{C})$ can be considered a measurable operator.

We propose the following theorem.
\begin{thm}(Hermitian decomposition of a complex matrix)\\
$\forall_{A \in  \mathbb{M}_n (\mathbb{C})}\text{ } \exists_{Q,R \in B_*(\mathcal{H})}$ such that the matrix $A$ can be decomposed as 
\begin{equation}
A = Q + i R.
\end{equation}
\end{thm}
\begin{proof}
Let us first denote $A=[a_{ij}]$ and since in general $a_{ij} \in \mathbb{C}$ we can put
\begin{equation}
a_{ij} = Re \text{ } a_{ij} + i Im \text{ } a_{ij}.
\end{equation}
Moreover we can denote $Q = [q_{ij}]$ and $R = [r_{ij}]$. Then we shall define the entries of the matrices $Q$ and $R$ in the way:
\begin{equation}
q_{ij} : = \frac{Re\text{ } a_{ij} + Re \text{ }a_{ji}}{2} + i \frac{Im \text{ }a_{ij} - Im \text{ } a_{ji}}{2},
\end{equation}
\begin{equation}
r_{ij} := \frac{Im \text{ } a_{ij} + Im \text{ } a_{ji}}{2} + i \frac{Re \text{ } a_{ji} - Re\text{ } a_{ij}}{2}.
\end{equation}
One can easily notice that $\overline{q_{ij}} = q_{ji}$ and $\overline{r_{ij}} = r_{ji}$. Therefore $Q,R \in B_*(\mathcal{H})$.

Furthermore, one can check that 
\begin{equation}
q_{ij} + i r_{ij} = \frac{Re\text{ } a_{ij} + Re \text{ }a_{ji}}{2} + i \frac{Im \text{ }a_{ij} - Im \text{ } a_{ji}}{2} + i \frac{Im \text{ } a_{ij} + Im \text{ } a_{ji}}{2} + \frac{Re \text{ } a_{ij} - Re\text{ } a_{ji}}{2} = a_{ij},
\end{equation}
which implies that 
\begin{equation}
A = Q + i R.
\end{equation}
\end{proof}

The above theorem states that every complex matrix $A\in \mathbb{M}_n (\mathbb{C})$ can be uniquely decomposed into two hermitian matrices. In other words, every complex matrix can be regarded as a pair of observables (hermitian matrices), i. e.
\begin{equation}
A \rightarrow (Q_1, Q_2), \text{  } Q_1,Q_2 \in B_* (\mathcal{H}).
\end{equation}

Since in general any observable is considered measurable, therefore, any complex matrix can also be considered a measurable operator.

In this paper it has been proven that for any $A\in \mathbb{M}_n (\mathbb{C})$ there exist two observables $Q_1,Q_2 \in B_* (\mathcal{H})$ such that 
\begin{equation}
A = Q_1 + i Q_2.
\end{equation}

If we generalize the idea of quantum measurement, we can define the mean value of the operator $A \in \mathbb{M}_n$ (and denote by $\langle A \rangle _t$) on a quantum system characterized by a density matrix $\rho (t) $ in the following way
\begin{equation}
\langle A \rangle _t := Tr [A \rho(t)] = Tr \left [ (Q_1 + i Q_2) \rho(t) \right ].
\end{equation}
Taking the advantage of the fact the \textit{trace} is linear one obtains
\begin{equation}
\langle A \rangle _t = Tr [Q_1 \rho(t) ] + i Tr [Q_2 \rho (t)],
\end{equation}
which can be equivalently presented as
\begin{equation}
\langle A \rangle _t = \langle Q_1 \rangle _t + i \langle Q_2 \rangle _t.
\end{equation}

One can observe that if we generalize the idea of quantum measurement in such a way that we treat any complex matrix $A\in \mathbb{M}_n (\mathbb{C})$ as a measurable operator, the mean value of $A$ is a complex number such that its real and imaginary parts are mean values of the observables $Q_1, Q_2$ that appear in the hermitian decomposition of $A$. Therefore, the measurement of any complex operator $A$ can be understood as the measurement of two physical quntities that are mathematically represented by two hermitian matrices $Q_1, Q_2$.

\section{Connection with the stroboscopic tomography}

When considering problems in the stroboscopic tomography, one needs to bear in mind the necessary condidion that the set of observables $Q_1, Q_2, ..., Q_r$ has to fulfill so that an open quantum system with dynamics given by \eqref{eq:kossak} will be reconstructible. 

\begin{thm}
An open quantum system which evolution is given by \eqref{eq:kossak} is $(Q_1,...Q_r)$-reconstructible if and only if the operators $Q_i$ satisfy the condition \cite{jam83,jam00}
\begin{equation}\label{eq:condition}
\bigoplus \limits_{i=0}^r K_\mu (\mathbb{L},Q_i) = B_*(\mathcal{H}),
\end{equation}
where $\bigoplus$ refers to the Minkowski sum, $\mu$ is the degree of the minimal polynomial of $\mathbb{L}$ and $K_\mu (\mathbb{L}, Q_i)$ denotes Krylov subspace which standard difinition reads
\begin{equation}
 K_\mu (\mathbb{L}, Q_i) := Span \{ Q_i, \mathbb{L}^* Q_i, (\mathbb{L}^*)^2 Q_i, ...,(\mathbb{L}^*)^{\mu-1} Q_i \}.
\end{equation}
\end{thm}

In reference to this condition for observability of a quantum system we can propose the following theorem.
\begin{thm}
Assume that the hermitian matrices $\{ \lambda_1, \lambda_2, ..., \lambda_{n^2}\}$ form a basis in the space of all hermitian operators $B_* (\mathcal{H})$, where $n = dim \mathcal{H}$. Then they also constitute a basis in the space of all linear operators $\mathbb{M}_n (\mathbb{C})$.
\end{thm}
\begin{proof}
Taking into account the assumption, $\forall_{Q \in B_*(\mathcal{H})}$ $\exists_{\alpha_1,...,\alpha_{n^2} \in \mathbb{R}}$ such that 
\begin{equation}\label{eq:QQQ}
Q = \sum_{k=1}^{n^2} \alpha_k \lambda_k.
\end{equation}
Then from the theorem on hermitian decomposition of a complex matrix it follows that $\forall_{A \in  \mathbb{M}_n (\mathbb{C})}\text{ } \exists_{Q,R \in B_*(\mathcal{H})}$ such that the matrix $A$ can be decomposed as 
\begin{equation}
A = Q + i R.
\end{equation}
Assuming that $Q$ has such decomposition as in \eqref{eq:QQQ} and taking $R$ in form
\begin{equation}
R = \sum_{k=1}^{n^2} \beta_k \lambda_k \text{,  } \beta_k \in \mathbb{R},
\end{equation}
matrix $A$ can be represented as
\begin{equation}
A = \sum_{k=1}^{n^2} \alpha_k \lambda_k + i \left (  \sum_{k=1}^{n^2} \beta_k \lambda_k \right ),
\end{equation}
which can be transformed into the form
\begin{equation}
A = \sum_{k=1}^{n^2} \left ( \alpha_k + i \beta_l \right ) \lambda_k.
\end{equation}
Finally, the matrix $A$ can be decomposed as
\begin{equation}\label{eq:final}
A = \sum_{k=1}^{n^2} z_k \lambda_k,
\end{equation}
where $z_k \in \mathbb{C}$ and $z_k = \alpha_k + i \beta_k$.\\
From the equation \eqref{eq:final} one can easily draw a conclusion that the set of matrices $\{ \lambda_1, \lambda_2, ..., \lambda_{n^2}\}$ is a basis in $\mathbb{M}_N (\mathbb{C})$.
\end{proof}

The link between the above theorem and the stroboscopic tomography is that in the equation \eqref{eq:condition}, which expresses the necessary condition for observability, on the right hand side you can put either $B_*(\mathcal{H})$ or $B (\mathcal{H})$. On the basis of the theorem 3.2 one can conclude that if certain operators span one of these spaces, they also have to span the other.

\section{Summary}

In this paper it has been proved that any complex matrix $A \in \mathbb{M}_n (\mathbb{C})$ can be uniquely determined by two hermitian matrices (i.e. observables). In general mean values of hermitian matrices can be obtained from an experiment. Thus from this observation one can conclude that any complex matrix can be regarded as a measurable operator. The measurement of a complex matrix should be understood as the measurement of the two observables that determine the complex matrix. The measurement result of a complex matrix is a complex number which real and imaginary parts are obtained from an experiment. Further research is planned to investigate whether the generalized approach to measurable operators can improve the models of the stroboscopic tomography.


\begin{thebibliography}{99}

\bibitem{altepeter04}
Altepeter J.~B., James D.~F.~V., Kwiat P.~G.: 4 Qubit Quantum State Tomography. In: Paris, M.~G.~A., Rehacek, J. (eds.) Quantum State Estimation, pp. 111-145. Springer, Berlin (2004)

\bibitem{alicki87}
Alicki R., Lendi K.: Quantum Dynamical Semigroups and Applications, Springer, Berlin (1987)

\bibitem{genki03}
Kimura G.: The Bloch vector for N-level systems. Phys. Lett. A {\em{314}}, 339-349 (2003)

\bibitem{bamber}
Lundeen J. S., Sutherland B., Patel A., Stewart C., Bamber C.: Direct measurement of the quantum wavefunction. Nature {\em{474}}, 188 (2011)

\bibitem{wu}
Wu S.: State tomography via weak masurements. Scientific reports {\em{3}}, 1193 (2013).

\bibitem{jam83}
Jamio\l{}kowski A.: The minimal Number of Operators for Observability of N-level Quantum Systems. Int. J. Theor. Phys. {\em{22}}, 369-376 (1983)

\bibitem{jam00}
Jamio\l{}kowski A.: On complete and incomplete sets of observables, the principle of maximum entropy—revisited. Rep. Math. Phys. {\em{46}}, 469-482 (2000)

\bibitem{jam04}
Jamio\l{}kowski A.: On a Stroboscopic Approach to Quantum Tomography of Qudits Governed by Gaussian Semigroups. OSID {\em{11}}, 63-70 (2004) 

\bibitem{gorini76}
Gorini, V., Kossakowski, A., Sudarshan, E.~C.~G.: Completely Positive Dynamical Semigroups of N-level Systems. J. Math. Phys. {\em{17}}, 821-825 (1976)

\end{thebibliography}
\end{document}